\documentclass[a4paper,oneside,11pt,reqno]{amsart}

\usepackage{amssymb}
\usepackage{amsthm}
\usepackage{amsmath}
\usepackage{url}
\usepackage{hyperref}
\usepackage{cite}
\usepackage{graphicx}
\usepackage[margin=1in]{geometry}
\usepackage[foot]{amsaddr}
\usepackage{comment}
\usepackage{xcolor}
\usepackage{enumitem}

\newtheorem{definition}{Definition}
\newtheorem{theorem}{Theorem}
\newtheorem{corollary}{Corollary}
\newtheorem{lemma}{Lemma}
\newtheorem{proposition}{Proposition}
\theoremstyle{remark}

\newcommand{\ea}[1]{\begin{align}#1\end{align}}
\newcommand{\nn}{\nonumber \\}
\newcommand{\la}{\langle}
\newcommand{\ra}{\rangle}
\newcommand{\mcl}{\mathcal}

\newcommand{\mbb}{\mathbb}
\newcommand{\bmt}{\begin{pmatrix}}
\newcommand{\emt}{\end{pmatrix}}

\DeclareMathOperator{\Tr}{Tr}

\begin{document}

\title{Geometric properties of SIC-POVM tensor square
}

\author{Vasyl Ostrovskyi \and Danylo Yakymenko}
\email{vo@imath.kiev.ua, danylo.yakymenko@gmail.com}  
\address{Institute of mathematics of the National Academy of Sciences of Ukraine, Kyiv, Ukraine}

\begin{abstract}
It's known that if $d^2$ vectors from $d$-dimensional Hilbert space $H$ form a SIC-POVM (SIC for short) then tensor square of those vectors form an equiangular tight frame on the symmetric subspace of $H\otimes H$ \cite{Renes}. We prove that for any SIC of WH-type (Weyl-Heisenberg group covariant) this squared frame can be obtained as a projection of a WH-type basis of $H\otimes H$ onto the symmetric subspace. We give a full description of the set of all WH-type bases, so this set could be used as a search space for SIC solutions. Also we show that a particular element of this set is close to a SIC solution in some structural sense. 
Finally we give a geometric construction of a SIC-related symmetric tight fusion frames that were discovered in odd dimensions in \cite{STFF}.
\end{abstract}

\maketitle

\section{Introduction}
\label{intro}

Let $H$ be a complex Hilbert space with finite dimension $d$. A \textit{symmetric, informationally complete, positive operator-valued measure} (\textit{SIC-POVM} or \textit{SIC} for short) is a set of $d^2$ rank-1 projectors $\Pi_i \in \mcl{L}(H)$, $i \in [1..d^2]$, $\Pi_i^2=\Pi_i^\dag=\Pi_i$, such that 
\ea{
	\sum_{i \in [1..d^2]} \frac{1}{d}\Pi_i = I_d ~\text{ and }~ \forall i,j \in [1..d^2] :~ \Tr(\Pi_i \Pi_j) = \frac{d\delta_{ij}+1}{d+1},
}

where $\delta_{ij}$ is the Kronecker delta.

The set of $d^2$ unit-norm vectors $|v_i\rangle$ from $H$ that correspond to those projectors, i.e. $|v_i\rangle \langle v_i| = \Pi_i$, also will be called SIC. Another terminology for such set of vectors is \textit{a maximal equiangular tight frame (maximal ETF)} on $H$. 

Although equiangular tight frames have a rich history of research, it's mostly in a real space. We refer to \cite{ETF} for a collection of known results about ETFs. 

The first major study of SICs was made by G. Zauner in his doctoral thesis devoted to quantum designs \cite{Zauner}. He constructed SICs in low dimensions, stated the conjecture about the existence of WH-type SICs (Weyl-Heisenberg group covariant) in every dimension $d$, which is now known as \textit{Zauner's conjecture}, and its stronger form -- the existence of WH-type SICs of a special kind (related to order 3 unitary symmetry). Note that Zauner's conjecture is in contrast to the real case where a maximal ETF of $d(d+1)/2$ unit-norm vectors rarely exists for a given $d$. 

The substantial attention this conjecture received after the work of J. M. Renes et al. \cite{Renes} where the authors independently from G. Zauner formulated the conjecture, constructed SICs in low dimensions, found approximate numerical SICs for dimensions up to 45 and related SICs to quantum spherical 2-designs. 

To this day the SIC existence conjecture has a strong supporting evidence. There are found exact solutions for all $d\leq 21$ and for a lot of other dimensions up to 2208, and approximate solutions for all $d \leq 151$. 
We refer to \cite{STFF}, \cite{Fuchs}, \cite{Waldron}, \cite{Scott} for a more detailed overview of history, significance, known results and references on this subject.

\section{Preliminaries and results}
\label{sec:2}

As $\mcl{L}(H)$ and $\mcl{U}(H)$ we denote the set of linear and unitary operators on $H$, respectively. 
The set $\{ |0\rangle, |1\rangle, .. , |d-1\rangle \}$ denotes an orthonormal basis in $H$. 

Let 
\ea{
	\tau = exp( 2\pi i \cdot \frac{d+1}{2d}), ~~~ \omega = \tau^2 = exp( 2\pi i \cdot \frac{1}{d} ),
}
so $\omega^d=\tau^{2d}=1$. Although $\tau^{d}=1$ if $d$ is odd. 

Consider the so-called \textit{clock and shift} unitary matrices $C,S \in \mcl{U}(H)$, defined by

\ea{
	\forall i \in [0..d-1]: ~~~ C |i\ra = \omega^i |i\ra, ~~~ S |i\ra = |i+1\ra,
}
where the sum in bra-kets is taken modulo $d$.

It can be seen that 
\ea{
	C^d = S^d = I ~\text{and}~ CS = \omega SC,
}
so the set $\{\omega^k C^iS^j \}_{k,i,j \in [0..d-1]}$ forms a group (generated by $C$ and $S$) of order $d^3$. Some authors already refer to this group as Weyl-Heisenberg group, though we will use a slightly different definition. Note that the set $\{ \frac{1}{\sqrt{d}}C^iS^j \}_{i,j \in [0..d-1]}$ forms an orthonormal basis for operator space $\mcl{L}(H)$ with the Hilbert-Schmidt inner product on it. 

It appears to be reasonable to consider the so-called \textit{displacement} (or \textit{translation}) unitary operators $T_k$, $\forall k \in \mbb{Z}^2$, defined by
\ea{
	T_k = T_{(k_1,k_2)} = \tau^{k_1 k_2} S^{k_1}C^{k_2},
}
so, in particular, $T_{(0,0)} = I$, $T_{(1,0)} = S$, $T_{(0,1)} = C$. 

We will also use the notation $\mbb{T}_k =  T_k \otimes T_k = T_k^{\otimes 2}$ throughout this paper. 

These operators satisfy the following properties \cite{Appleby}: 

$\forall a,b \in \mbb{Z}^2$:
\ea{
	T_a^\dag = T_{a}^{-1} = T_{-a},
	\\
	T_a T_b = \tau^{<a,b>} T_{a+b},
}
where the symplectic form $<a,b>$ is defined by
\ea{
	<a,b> = a_2b_1 - b_2a_1 = - <b,a>.
}

Operators $T_a$, $a \in \mbb{Z}^2$, are periodic in $d \times d$ lattice if $d$ is odd and in $2d \times 2d$ lattice if $d$ is even. That is
\ea{
	T_{a+db} = \epsilon^{<a,b>} T_{a},
}
where 
\ea{
	\epsilon = \tau^d = \begin{cases} 1, \text{ if } d \text{ is odd } \\ -1, \text{ if } d \text{ is even } \end{cases}
}

The set $\{T_a\}_{a\in \mbb{Z}^2}$ generates a group which we refer to as Weyl-Heisenberg group. 
This group is also generated by $\tau, C, S$. It has order $d^3$ if $d$ is odd and order $2d^3$ if $d$ is even. 

Note that the set $\{ \frac{1}{\sqrt{d}} T_{a} \}_{a \in [0..d-1]^2} $ is also an orthonormal basis for operator space $\mcl{L}(H)$. So the following proposition holds \cite{Renes}:
\begin{proposition}\label{prop:1}
For every unit-norm vector $|f\ra \in H$ the set $\{ T_i |f\ra\}_{i \in [0..d-1]^2}$ is a tight frame on $H$. That is 
\ea{
	\sum_{i \in [0..d-1]^2} T_i |f \ra  \la f| T_{-i} = d I
}
\end{proposition}

Let $|f\rangle \in H$ be a unit-norm vector. The set $ \{ T_a |f\rangle\}_{a \in [0..d-1]^2}$ is called \textit{SIC of WH-type} if it's a SIC, that is $| \langle f | T_a | f \rangle | ^2 = \frac{1}{d+1}$ if $a \neq (0,0)$. In this case $|f\rangle$ is called \textit{fiducial vector of a SIC of WH-type}.

With the single exception of the Hoggar lines \cite{Hoggar} in dimension 8, every known SIC is of WH-type. 

As $H_{sym}$ we denote the symmetric subspace of $H \otimes H$, and $P_{sym}$ is the projector on $H_{sym}$ (see \cite{SymmChurch} for a review of the symmetric subspace properties). Also we denote $H_{asym} = H_{sym}^\perp$ and $P_{asym} = I - P_{sym}$. Note that $\text{dim}(H_{sym}) = \Tr(P_{sym}) = \frac{d(d+1)}{2}$. 

In this paper our attention is focused on the following properties of SICs, see \cite{Renes}.

 \begin{proposition}\label{prop:2}
Let $\{ |f_i\ra \in H \}_{i \in [1..d^2]}$ be a SIC in $H$. Then $\{ |f_i\ra|f_i\ra \in H \otimes H\}_{i \in [1..d^2]}$ is an equiangular tight frame on $H_{sym}$, that is 
\ea{
	\frac{1}{d^2} \sum_{i \in [1..d^2]} |f_i\ra |f_i\ra \la f_i | \la f_i | = \frac{2}{d(d+1)} P_{sym}
}
\end{proposition}

Proposition \ref{prop:2} states that the tensor square of a SIC is an equiangular tight frame on the symmetric subspace of $H \otimes H$. 
Note that there exists a multi-parametric family of ETFs of size $d^2$ in dimension $\frac{d(d+1)}{2}$ -- the dimension of the symmetric subspace, see \cite{Szollosi}. Though it's unlikely that some of those ETFs are obtained from a SIC tensor square. 

What's encouraging is that proposition \ref{prop:2} is almost sufficient for a SIC characterization due to the following 

\begin{proposition}\label{prop:3}
Suppose that the set of unit-norm vectors $\{ |f_i\ra \}_{i \in [1..d^2] } \subset H$ forms a tight frame in $H$ and also $\{ |f_i\ra|f_i\ra \}_{i \in [1..d^2] } \subset H\otimes H $ forms a tight frame on $H_{sym}$. Then $\{ |f_i\ra \}_{i \in [1..d^2] }$ is a SIC. 
\end{proposition}

Proposition \ref{prop:3} is just a restatement of the equivalence between SICs and spherical 2-designs obtained in \cite{Renes}. 

Together, propositions \ref{prop:1},\ref{prop:2} and \ref{prop:3} imply the following WH-type SIC criterion

\begin{theorem}\label{tm:sic-sym-crit}
The unit-norm vector $ |f\ra \in H$ is a fiducial vector of a SIC of WH-type if and only if  
\ea{
	\frac{1}{d^2} \sum_{i \in [0..d-1]^2} \mbb{T}_i |f \ra |f\ra \la f | \la f | \mbb{T}_{-i} = \frac{2}{d(d+1)} P_{sym}
} 
\end{theorem}

This criterion is essentially the main starting point of our research. \\

By Naimark's theorem \cite{Casazza} every Parseval frame of size $n$ in dimension $m$ is unitary equivalent to a projection of some orthonormal basis of $n$-dimensional space onto its subspace of dimension $m$. Note that if $\{ |v_i\ra \} \in H$ is a tight frame with the frame bound $B$ then $\{ \frac{1}{\sqrt{B}}|v_i\ra \} $ is a Parseval frame. Thus we can deduce the following corollary 

\begin{corollary}
Let $ |f\ra \in H$ be a fiducial vector of a SIC of WH-type. Then there exists an orthonormal basis $\{ |b_i\rangle \}_{i\in[0..d-1]^2}$ of $H \otimes H$ such that $\forall i\in[0..d-1]^2$: 
\ea{
	\sqrt{\frac{d+1}{2d}}\mbb{T}_i |f \ra |f\ra = P_{sym} | b_i \rangle
}
\end{corollary}

For a particular SIC such basis $\{ | b_i \rangle \}$ is not unique. A rather natural question arise. Can we find such basis with some additional nice structure? In particular, can it be WH group covariant? It turns out that the answer is yes. 

In this paper we introduce the notion of WH-type basis of $H\otimes H$ (def. \ref{def:wh-basis}), describe the set of all such bases (theorems \ref{tm:bwh-odd}, \ref{tm:bwh-even}), and prove that for every SIC its tensor square is a projection of some WH-type basis onto $H_{sym}$ (theorem \ref{tm:main}). So the set of WH-type bases can be used as a search space for SIC solutions, see section \ref{sec:search-space}. Also we show that a particular WH-type basis is close to a SIC solution in some structural sense (prop. \ref{prop:0F0}). Finally, as a collateral result, we found a geometric construction of a SIC-related symmetric tight fusion frames (STFFs) of ranks $(d\pm1)/2$ in dimension $d$ constructed for odd $d$ in \cite{STFF}, see section \ref{sec:STFF}. 

\section{A note on the Schmidt decomposition of a symmetric vector}
\label{sec:srank}

In this paper we work a lot with the symmetric subspace so it's useful to keep in mind the following property of a symmetric vectors. 

Recall that the Schmidt rank of a vector $|v\ra \in H\otimes H$, which we denote as $srank(|v\ra)$, is the number of non-zero coefficients in the Schmidt decomposition of $|v\ra$. That is $|v\ra = \sum_{i=1}^r c_i |a_i\ra \otimes |b_i\ra$, $c_i>0$, $\la a_k | a_l \ra = \la b_k | b_l \ra = \delta_{kl}$, $r = srank(|v\ra)$.  However, for a vector from the symmetric subspace $|v\ra \in H_{sym}$ one can define the $symmetric ~rank$ as the minimum number $r$ such that $|v\ra = \sum_{i=1}^r c_i |a_i\ra \otimes |a_i\ra$, where $|a_i\ra$ are not necessary orthogonal to each other and $c_i$ are any complex numbers. But in fact, this rank coincides with the Schmidt rank. 

\begin{theorem}\label{tm:symm}
For a vector $|v\ra \in H_{sym}$ its symmetric rank coincides with its Schmidt rank. Moreover, if $r = srank(|v\ra)$, then there exists a symmetric Schmidt decomposition
\ea{
	|v\ra = \sum_{i=1}^r c_i |a_i\ra \otimes |a_i\ra,
}
where $c_i >0$, $\la a_k | a_l \ra = \delta_{kl}$. If we rearrange summands such that $c_i \le c_j$ for $i < j$ then this decomposition is unique up to a real orthogonal matrix $U$ that preserves corresponding diagonal matrix $D = diag\{c_1,..,c_r\}$, that is $D = U^TDU$.

\end{theorem}

\begin{proof}
Let $|v\ra = \sum_{i=1}^r c_i |a_i\ra \otimes |a_i\ra$ for some complex $c_i$ and some vectors $|a_i\ra \in H$. Recall that the Schmidt rank of $|v\ra$ is the same as the rank of the partial trace (over any subsystem) of $|v\ra\la v|$. Let's calculate 
\ea{
\text{Tr}_1 |v\ra\la v| 
= \text{Tr}_1 \bigg( \sum_{i=1}^r c_i |a_i\ra \otimes |a_i\ra \sum_{j=1}^r \bar{c_i} \la a_i| \otimes \la a_i | \bigg)
=
\nn
= \text{Tr}_1 \bigg( \sum_{i=1}^r  \sum_{j=1}^r   c_i \bar{c_j} |a_i\ra \la a_j| \otimes |a_i\ra \la a_j| \bigg) 
= \sum_{i=1}^r  \sum_{j=1}^r   c_i \bar{c_j} \Tr \big( |a_i\ra \la a_j| \big)  |a_i\ra \la a_j|
}

It's easy to see that this operator preserves subspace $H_a = span\{ |a_i\ra \} \subset H$ and it's equal to 0 on $H_a^\perp$. Hence its rank can't be more than $r$. This proves that Schmidt rank is not bigger than symmetric rank. 

Now let $| v \ra \in H_{sym}$ has the Schmidt rank $r$. That is $|v\ra = \sum_{i=1}^r c_i |a_i\ra \otimes |b_i\ra$, $c_i>0$, $\la a_k | a_l \ra = \la b_k | b_l \ra = \delta_{kl}$. For a symmetric vector both reduced states of $| v \ra \la v |$ coincide, that is $\text{Tr}_1 |v\ra\la v| = \text{Tr}_2 |v\ra\la v|$. This implies 
\ea{
	\sum_{i=1}^r  c_i^2 |a_i\ra \la a_i | =  \sum_{i=1}^r  c_i^2 |b_i\ra \la b_i |, 
}
hence $span\{ |a_i\ra \} = span\{ |b_i\ra \} $. If $c_i$ were all different we would have $|a_i\ra \la a_i | = |b_i\ra \la b_i |$ and a required symmetric decomposition follows easily. 
In general situation we can represent vectors $|b_i\ra$ as linear combinations of vectors from the set $\{|a_i\ra\}$. Hence vector $| v \ra$ can be represented as 
\ea{ 
	| v \ra = \sum_{i=1}^r  \sum_{j=1}^r   c_{ij} |a_i\ra | a_j \ra 
}
for some complex numbers $c_{ij}$. Note that $\{ |a_i\ra | a_j \ra \}_{i,j \in [1..r]}$ is a basis of a subspace of rank $r^2$. Since $| v \ra$ is symmetric we have $U_{swap} | v \ra = | v \ra$, where $U_{swap}| \phi \ra | \psi \ra = | \psi \ra | \phi \ra$ is the swap operator. This implies that $c_{ij} = c_{ji}$, $\forall i,j \in [1..r]$. So the corresponding $r \times r$ matrix $M$ with entries $c_{ij}$ is symmetric. By Autonne-Takagi factorization \cite{HornJohnson} there is a unitary $U$ (with the entries $u_{ij}$) such that $U^TMU = D$, where $D$ is the diagonal matrix with a real non-negative entries. Hence substitutions $|a_i\ra = \sum_{j=1}^r u_{ij}|a_j^\prime\ra$ give a symmetric decomposition (which is also a Schmidt decomposition) -- with the symmetric rank not bigger than $r$.
 
Finally, let $|v\ra = \sum_{i=1}^r \alpha_i |a_i\ra |a_i\ra = \sum_{i=1}^r \beta_i |b_i\ra |b_i\ra$, $0 < \alpha_i \le \alpha_j, 0 < \beta_i \le \beta_j$, for $i <j$, and $\la a_k | a_l \ra = \la b_k | b_l \ra = \delta_{kl}$. Again we can deduce that $span\{ |a_i\ra \} = span\{ |b_i\ra \} $ so there is a unitary $U$ with the entries $u_{ij}$ such that $|b_i\ra = \sum_{j=1}^r u_{ij}|a_j\ra$. If we denote $A = diag\{\alpha_1,..,\alpha_r\}$ and $B = diag\{\beta_1,..,\beta_r\}$ then we have 
\ea{ 
	A = U^T B U
}
Hence $ \bar{U} A = B U$ which gives $\bar{u}_{ij} \alpha_j = u_{ij} \beta_i$, $\forall i,j \in [1..r]$. This implies $u_{ij}=0$ if $\alpha_j \neq \beta_i$ and $u_{ij} \in \mbb{R}$ if $\alpha_j = \beta_i$. Hence $U$ is a real orthogonal matrix. So it must be $A=B$. Also we can deduce that $U$ is block-diagonal where the blocks correspond to a subsets of indices where the numbers $\alpha_i, \alpha_{i+1}, .. , \alpha_{i+t} $ are equal. 

\end{proof}

\section{WH-type bases of $H\otimes H$}
\label{sec:whb}

\begin{definition}\label{def:wh-basis}
A WH-type basis is an orthonormal basis of $H\otimes H$ of the form \\
$\{ \mbb{T}_i |b\ra \}_{i \in [0..d-1]^2}$. 
That is $| b \ra \in H\otimes H$, $\la b|b\ra=1$ and $\forall i \in [0..d-1]^2, i\neq (0,0)$: 
\ea{
	\la b| \mbb{T}_i | b \ra = 0
}
Vector $|b\ra $ is called a fiducial basis vector. 

As $\mbb{B}_{WH}$ we denote the set of fiducial basis vectors of all possible WH-type bases.
\end{definition}

Our next aim is to describe the set $\mbb{B}_{WH}$. 

First of all, let's show that this set is not empty. Consider the vector $|0\ra F |0\ra$:
\ea{
	|0\ra F |0\ra = |0\ra \otimes F |0\ra = \frac{1}{\sqrt{d}} \sum_{i \in [0..d-1]} |0i\ra
}

Here $F$ is the discrete Fourier transform matrix 
\ea{
	F = \frac{1}{\sqrt{d}} \sum_{i,j \in [0..d-1]} \omega^{ij} |i\ra\la j| 
}

Matrix $F$ is an element of the automorphism group of the WH group since it interchanges the shift and clock matrices, that is $F^{-1}CF=S, ~ F^{-1}SF=C^{-1}$. In general it acts on the WH group by the following formula
\ea{ 
	F^{-1} T_{(i_1,i_2)} F = T_{(i_2,-i_1)} 
}

\begin{proposition} \label{prop:0F0-bwh}
\ea{
	|0\ra F |0\ra \in \mbb{B}_{WH}
}
\end{proposition}
\begin{proof}
Let's calculate $\mbb{T}_i |0\ra F |0\ra$:
\ea{
	\mbb{T}_i |0\ra F |0\ra = T_i \otimes T_i  \cdot |0\ra \otimes F |0\ra = T_i |0\ra \otimes T_iF|0\ra =  
	\nn
	= T_{(i_1,i_2)} |0\ra \otimes FT_{(i_2,-i_1)}|0\ra
	=  |i_1\ra \otimes F |i_2\ra = | i_1\ra F |i_2\ra
}
It's easy to see that $\{| i_1\ra F |i_2\ra\}_{(i_1,i_2) \in [0..d-1]^2}$ is a basis of $H \otimes H$, hence $|0\ra F |0\ra \in \mbb{B}_{WH}$. 
\end{proof}

What is fascinating, as we will see in prop. \ref{prop:0F0}, is that this basis $\{| i_1\ra F | i_2\ra\}$, when projected onto the symmetric subspace $H_{sym}$, gives an equiangular tight frame (on the symmetric subspace) such that each vector from it has the Schmidt rank 2 ! If instead this Schmidt rank were 1, then by theorem \ref{tm:sic-sym-crit} we would have a SIC solution (see section \ref{sec:search-space} for more details). Note that general vector in $H \otimes H$ has the Schmidt rank $d$. So in some structural sense $|0\ra F |0\ra$ generates a structure which is close to a SIC solution. 

Let's go to the details. Since $\mbb{T}_i$ commutes with $P_{sym}$ we can deduce the following 

\begin{proposition}
Let $| b \ra \in \mbb{B}_{WH}$. Then every element of $\{ P_{sym}\mbb{T}_i |b\ra \}_{i \in [0..d-1]^2}$, has the same Schmidt rank. 
\end{proposition}

\begin{proof}
Let $|v\ra = P_{sym} |b\ra$ and $r = srank(| v \ra)$. By \ref{tm:symm} we have $|v\ra = \sum_{i=1}^r c_i |a_i\ra |a_i\ra$, where $c_i >0$, $\la a_k | a_l \ra = \delta_{kl}$. Then $P_{sym} \mbb{T}_k |b\ra = \mbb{T}_k P_{sym} |b\ra = \mbb{T}_k |v\ra = \sum_{i=1}^r c_i T_k |a_i\ra \otimes T_k |a_i\ra$. 

\end{proof}

\begin{proposition}
\ea{
	srank( P_{sym} |0\ra F|0\ra) = 2
}
\end{proposition}
\begin{proof}
Let $|v\ra = P_{sym} |0\ra F|0\ra$. 
$$ 
	|v\ra = P_{sym} \frac{1}{\sqrt{d}} \sum_{i =0}^{d-1} |0i\ra = P_{sym} \frac{1}{\sqrt{d}} \left( |00\ra + \frac{1}{2}\sum_{i=1}^{d-1} (|0i 	\ra + |i0 \ra) + \frac{1}{2}\sum_{i=1}^{d-1} (|0i \ra - |i0 \ra) \right) = 
$$
\ea{
	= \frac{1}{\sqrt{d}} \left( |00\ra + \frac{1}{2}\sum_{i=1}^{d-1} (|0i \ra + |i0 \ra) \right)
}

Let's compute $\Tr_1 |v\ra\la v|$:
$$
	\Tr_1 |v\ra\la v| =  \Tr_1 \frac{1}{d}  \left( |00\ra + \frac{1}{2}\sum_{i=1}^{d-1} (|0i \ra + |i0 \ra) \right) \left( \la 00| + \frac{1}{2}\sum_{j=1}^{d-1} (\la 0j | + \la j0 |) \right) = 
$$
$$
	= \frac{1}{d} \Tr_1 \Bigg( |00\ra\la 00| + \frac{1}{2}\sum_{j=1}^{d-1} (|00\ra\la 0j | + |00\ra\la j0 |) + \frac{1}{2}\sum_{i=1}^{d-1} (|0i \ra\la 00| + |i0 \ra\la 00 |) + 
$$
$$
+  \frac{1}{4}\sum_{i=1}^{d-1}\sum_{j=1}^{d-1} (|0i \ra \la 0j | + |0i \ra \la j0 | + | i0 \ra \la 0j | + | i0 \ra \la j0 | ) \Bigg) = 
$$
$$
	 = \frac{1}{d} \Bigg( |0\ra\la 0| + \frac{1}{2}\sum_{j=1}^{d-1} (|0\ra\la j |) + \frac{1}{2}\sum_{i=1}^{d-1} (|i\ra\la 0 |) + 
 \frac{1}{4}\sum_{i=1}^{d-1}\sum_{j=1}^{d-1} |i \ra \la j | +  \frac{d-1}{4}|0 \ra \la 0 | \Bigg) = 
$$
\ea{
	 = \frac{1}{d} \Bigg( \frac{d+3}{4} |0\ra\la 0| + \frac{1}{2}\sum_{i=1}^{d-1} (|0\ra\la i | + |i\ra\la 0 |) + 
 \frac{1}{4}\sum_{i=1}^{d-1}\sum_{j=1}^{d-1} |i \ra \la j |  \Bigg) 
}
It's easy to see that this is a real symmetric matrix that has the same $d-1$ last rows that are not proportional to the first row. So its rank is 2. 
\end{proof}

\begin{proposition}\label{prop:bwh-sym-proj}
Let $| b \ra \in \mbb{B}_{WH}$. Then $\{ \sqrt{\frac{2d}{d+1}}  P_{sym}\mbb{T}_i |b\ra \}_{i \in [0..d-1]^2}$ is the unit-norm tight frame on the symmetric subspace $H_{sym}$
\end{proposition} 
\begin{proof}
Note that the norm of every $P_{sym}\mbb{T}_i |b\ra$ is the same since $\mbb{T}_i$ commutes with $P_{sym}$.
It's a tight frame as a projection of a basis. That is we have 
\ea{
	\sum_{i \in [0..d-1]^2} P_{sym}\mbb{T}_i |b\ra \la b|  \mbb{T}_i^\dag P_{sym} = P_{sym}
}
By calculating the trace we conclude that $ d^2 \Tr ( P_{sym} |b\ra\la b | P_{sym} ) = \frac{d(d+1)}{2}$ hence the norm of $P_{sym} |b\ra$ equals to $\sqrt{\frac{d+1}{2d}}$.
\end{proof}

\begin{proposition}{(Combined properties of the frame generated by $|0\ra F|0\ra$)\\}
\label{prop:0F0}
Let $|b\ra = |0\ra F|0\ra $. The set $\{ \sqrt{\frac{2d}{d+1}} P_{sym}\mbb{T}_i |b\ra \}_{i \in [0..d-1]^2}$ is the unit-norm equiangular tight frame on the symmetric subspace $H_{sym}$ with each vector having Schmidt rank 2.
\end{proposition}
\begin{proof}
Let's prove the remaining equiangular property. 
It's convenient  to use the following equality
\begin{lemma}\label{lem:Psym}
\ea{
	P_{sym} = \frac{d+1}{2d}\Bigg( I + \frac{1}{d+1} \sum_{a \in [0..d-1]^2}^{a\neq (0,0)} T_a \otimes T_{-a} \Bigg)
}
\end{lemma}
\begin{proof}

Let's denote the right hand side of the equality as $P$ and let's prove that $P = P_{sym}$. It's easy to see that $P = P^\dag$ and $\Tr(P) = \frac{d(d+1)}{2}$. Now let's verify $\forall i \in [0..d-1]$:
$$
	P |i\rangle|i\rangle = \frac{d+1}{2d}\Bigg( I + \frac{1}{d+1} \sum_{a \in [0..d-1]^2}^{a\neq (0,0)} T_a \otimes T_{-a} \Bigg) |i\rangle|i\rangle = 
$$
$$ 
	= \frac{d+1}{2d}\Bigg( |i\rangle|i\rangle + \frac{1}{d+1} \sum_{a \in [0..d-1]^2}^{a\neq (0,0)} \tau^{2a_1a_2} |i+a_1\rangle|i - a_1\rangle \Bigg) =
$$
$$ 
	=  \frac{d+1}{2d}\Bigg( |i\rangle|i\rangle + \frac{1}{d+1} \sum_{a_1 \in [1..d-1]} |i+a_1\rangle|i - a_1\rangle \sum_{a_2 \in [0..d-1]} \omega^{a_1a_2}  +
$$
\ea{
	+ \frac{1}{d+1} \sum_{a_2 \in [1..d-1]} |i\rangle|i\rangle \Bigg) = \frac{d+1}{2d}\Bigg( |i\rangle|i\rangle + \frac{1}{d+1} (d-1)|i\rangle|i\rangle \Bigg) = |i\rangle|i\rangle
}

Similarly $\forall i,j \in [0..d-1], i \neq j$:

$$
P |i\rangle|j\rangle = \frac{d+1}{2d}\Bigg( I + \frac{1}{d+1} \sum_{a \in [0..d-1]^2}^{a\neq (0,0)} T_a \otimes T_{-a} \Bigg)|i\rangle|j\rangle = 
$$
$$
= \frac{d+1}{2d}\Bigg( |i\rangle|j\rangle + \frac{1}{d+1} \sum_{a \in [0..d-1]^2}^{a\neq (0,0)} \omega^{a_2(a_1+i-j)}|i+a_1\rangle|j-a_1\rangle \Bigg) = 
$$
$$
= \frac{d+1}{2d}\Bigg( |i\rangle|j\rangle + \frac{1}{d+1} \sum_{a_1 \in [1..d-1]}^{a_1\neq j-i} |i+a_1\rangle|j-a_1\rangle \sum_{a_2 \in [0..d-1]} \omega^{a_2(a_1+i-j)} +  
$$
$$
+ \frac{1}{d+1} |i\rangle|j\rangle \sum_{a_2 \in [1..d-1]} \omega^{a_2(0+i-j)} + \frac{1}{d+1} |j\rangle|i\rangle \sum_{a_2 \in [1..d-1]} \omega^{a_2((j-i)+i-j)}\Bigg) = 
$$
\ea{
	= \frac{d+1}{2d}\Bigg(|i\rangle|j\rangle - \frac{1}{d+1} |i\rangle|j\rangle + \frac{d}{d+1} |j\rangle|i\rangle \Bigg) = \frac{1}{2}(|i\rangle|	j\rangle + |j\rangle|i\rangle)
}
Hence $P (|i\rangle|j\rangle + |j\rangle|i\rangle) = |i\rangle|j\rangle + |j\rangle|i\rangle $ and $P (|i\rangle|j\rangle - |j\rangle|i\rangle) = 0$. So $P$ is identity on $H_{sym}$ and 0 on $H_{asym}$ therefore we can conclude that $P = P_{sym}$.

\end{proof}


Let's continue the proof of prop. \ref{prop:0F0}. 

Note that $\forall x = (x_1, x_2), y = (y_1, y_2) \in [0..d-1]^2$ :

$$
\la 0 | \la 0 | F^\dag \cdot T_x \otimes T_y \cdot | 0 \ra F |0\ra = \la 0 | T_{(x_1,x_2)} | 0 \ra \la 0 | T_{(y_2,-y_1)} | 0 \ra = 
$$
\ea{
	= \begin{cases} 1, & \text{ if } x_1=0 \text{ and } y_2=0 \\ 0, & \text{else} \end{cases}
}

Finally, for $k \in [0..d-1]^2$, $k \neq (0,0):$ 
$$
(P_{sym} |0\ra F |0 \ra)^\dag (P_{sym} \mbb{T}_k |0\ra F |0 \ra) = \la 0 | \la 0 | F^\dag \cdot P_{sym} \mbb{T}_k \cdot |0\ra F |0 \ra = 
$$
$$
= \la 0 | \la 0 | F^\dag \cdot \frac{d+1}{2d}( I + \frac{1}{d+1} \sum_{i \in [0..d-1]^2}^{i\neq (0,0)} T_i \otimes T_{-i} ) \cdot T_k \otimes T_k \cdot |0\ra F |0 \ra = 
$$
$$
= \la 0 | \la 0 | F^\dag \cdot \frac{d+1}{2d}( T_k \otimes T_k + \frac{1}{d+1} \sum_{i \in [0..d-1]^2}^{i\neq (0,0)} T_{i+k} \otimes T_{-i+k} ) \cdot |0\ra F |0 \ra = 
$$
\ea{
	= \la 0 | \la 0 | F^\dag \cdot \frac{1}{2d}( \sum_{i \in [0..d-1]^2}^{i\neq (0,0)} T_{i+k} \otimes T_{-i+k} ) \cdot |0\ra F | 0 \ra = \frac{1}{2d}, 
}
because there is only one $i$ such that $i_1+k_1 = 0$ and $-i_2+k_2 = 0$. 
So the modulus of a scalar product between frame vectors is $\frac{1}{2d} \cdot \frac{2d}{d+1} = \frac{1}{d+1}$. 
End of proof of prop. \ref{prop:0F0}.

\end{proof}

Let's go back to our aim -- the description of $\mbb{B}_{WH}$.

\begin{theorem}\label{tm:uwh-group} The following holds: \\
\begin{enumerate}
\item Let $|b_1\ra, |b_2\ra \in \mbb{B}_{WH}$. Then there is a unique unitary operator $U \in \mcl{U}(H\otimes H)$ such that $U|b_1\ra = |b_2\ra$ and $U = \sum_{i\in [0..d-1]^2} c_i \mbb{T}_i$ for some $c_i \in \mbb{C}$. It can be computed by the formula  
\ea{
U = \sum_{i\in [0..d-1]^2} \la b_1| \mbb{T}_i^\dagger | b_2 \ra \mbb{T}_i
}
\item Let $|b\ra \in \mbb{B}_{WH}$ and a unitary $U = \sum_{i\in [0..d-1]^2} c_i \mbb{T}_i$, $c_i \in \mbb{C}$. Then $U |b\ra \in \mbb{B}_{WH}$.
\end{enumerate}
\end{theorem}

Clearly the set of unitary operators $U = \sum_{i\in [0..d-1]^2} c_i \mbb{T}_i, c_i \in \mbb{C}$, forms a group. We will denote this group  as $\mbb{U}_{WH}$. By theorem \ref{tm:uwh-group} this group acts regularly on $\mbb{B}_{WH}$. As a corollary we can describe $\mbb{B}_{WH}$ in terms of this group. 

\begin{corollary} 
The set $\mbb{B}_{WH}$ is exactly the set $\{ U |0\ra F|0\ra \}_{U \in \mbb{U}_{WH}}$ 
\end{corollary} 

This description is not very useful yet. Also we'll not present the direct proof of theorem \ref{tm:uwh-group} here. Despite historically it was obtained first, in this paper we present a bit shorter proof of the main result -- the concrete description of the set $\mbb{B}_{WH}$  (theorems \ref{tm:bwh-odd}, \ref{tm:bwh-even}). It's also easier to deduce theorem \ref{tm:uwh-group} from the concrete description. 


Anyway, let's study the group $\mbb{U}_{WH}$. Consider the operators $C \otimes C$ and $S \otimes S$. It's easy to see that 
\ea{
	(C \otimes C) (S \otimes S) = \omega^2 (S \otimes S) (C \otimes C)
}

Recall that the representation theory of the equation 
\ea{
	AB = \lambda BA,  \text{ where }  \lambda^n = 1,  ~ \lambda = exp( 2\pi i \cdot \frac{k}{n}),  ~ gcd(k,n) = 1,
}
is a known subject that was studied in many contexts. We refer to \cite{Davidson}, theorem VII.5.1. It follows that every finite-dimensional representation of such pair $\{A,B\}$ is a direct sum of irreducible representations and for every $(\alpha, \beta) \in \mbb{C}^2$  there is a unique (up to unitary equivalence) irreducible representation in dimension $n$ such that $A^n = \alpha I, B^n = \beta I$. It can be constructed as $A = \sqrt[n]{\alpha}C^k, B = \sqrt[n]{\beta}S$ (where $C$ and $S$ are clock and shift matrices of size $n \times n$).  

From this point it can be seen that the cases when $d$ is odd and $d$ is even differ significantly. The number $\omega^2$ is a primitive root of 1 of order $d$ if $d$ is odd but $\omega^2$ is a primitive root of 1 of order $\frac{d}{2}$ if $d$ is even. Furthermore, if $d$ is odd then $(C \otimes C)^d = (S \otimes S)^d = I$, so $(\alpha,\beta) = (1,1)$, i.e. there is a unique irreducible representation $\pi = \{ C^2, S \}$ of the equation $AB=\omega^2BA$ in dimension $d$ and $\{C \otimes C, S \otimes S \}$ is a direct sum of $d$ irreducible representations equivalent to $\pi$. But if $d$ is even then $\text{Spectrum}(C \otimes C)^{\frac{d}{2}} = \text{Spectrum}(S \otimes S)^{\frac{d}{2}} = \{+1,-1\}$. In fact, 4 options possible for $(\alpha, \beta) = (\pm1, \pm 1)$ and there are 4 possible inequivalent irreducible representations $\pi_{(\pm1, \pm 1)}$ of dimension $\frac{d}{2}$ (which are $\{C,S\}, \{\omega C,S\}, \{C,\omega S\}, \{\omega C, \omega S\}$) in the decomposition of $\{C \otimes C, S \otimes S \}$. As we'll see later, every irreducible $\pi_{(\pm1, \pm 1)}$ appears exactly $\frac{d}{2}$ times in this decomposition. Let's go to the details. 

\subsection{The case when $d$ is odd} 

Since $\{C \otimes C, S \otimes S \}$ is equivalent to the direct sum of $d$ equivalent representations, then the decomposition of $H \otimes H$ onto invariant subspaces is not unique (in contrast to the case of the sum of inequivalent representations). To be precise we consider a one concrete decomposition of $H \otimes H$ onto $d$ invariant $d$-dimensional parts which is aligned with the decomposition onto $H_{sym} \oplus H_{asym}$.  

For $i=0$ denote 
\ea{
	H_{0} = span\{ |j\ra|j\ra \}_{ j\in [0..d-1] } = span\{ |0\ra|0\ra, |1\ra|1\ra , .. , |d-1\ra | d-1\ra \},
}
and for $i \in [1..\frac{d-1}{2}]$ denote
$$
H_{i} = span\{ \frac{1}{\sqrt{2}}(|j\ra|j + i\ra + |j +i \ra|j \ra) \}_{j\in [0..d-1]}= 
$$
$$
= span\{ \frac{1}{\sqrt{2}}(|0\ra|i\ra + |i \ra|0 \ra), \frac{1}{\sqrt{2}}(|1\ra|1 + i\ra + |1 +i \ra|1 \ra), .. 
$$
\ea{
	.., \frac{1}{\sqrt{2}}(|d-1\ra|d-1 + i\ra + |d-1 +i \ra|d-1 \ra)\},
}
$$
H_{-i} =  span\{ \frac{1}{\sqrt{2}}(|j\ra|j + i\ra - |j +i \ra|j \ra) \}_{j\in [0..d-1]}= 
$$
$$
= span\{ \frac{1}{\sqrt{2}}(|0\ra|i\ra - |i \ra|0 \ra), \frac{1}{\sqrt{2}}(|1\ra|1 + i\ra - |1 +i \ra|1 \ra), .. 
$$
\ea{
	.. , \frac{1}{\sqrt{2}}(|d-1\ra|d-1 + i\ra - |d-1 +i \ra|d-1 \ra)\},	
}
where the sum of indices in bra-kets is taken modulo $d$.

Clearly, $\text{dim} H_i = d$ and
\ea{
	\bigoplus_{i \in [0..\frac{d-1}{2}]} H_i = H_{sym},
}
\ea{
	\bigoplus_{i \in [1..\frac{d-1}{2}]} H_{-i} = H_{asym}
}

It's not hard to check that for all $i \in [-\frac{d-1}{2} .. \frac{d-1}{2}]$ subspaces $H_i$ are invariant subspaces for operators $C \otimes C$ and $S \otimes S$. So if 
\ea{
	U = \sum_{i,j\in [0..d-1]} c_{ij} (S\otimes S)^i(C \otimes C)^j
} 
then we can write 
\ea{ 
	U = \bigoplus_{i \in [-\frac{d-1}{2}, .. ,\frac{d-1}{2}]} U |_{H_i}
}
Furthermore, since all irreducible parts of $\{C \otimes C, S \otimes S \}$ are equivalent to $\{C^2, S\}$, then $\forall i \in [-\frac{d-1}{2} .. \frac{d-1}{2}]$ there are isometries $O_i: H \rightarrow H_i$ such that 
\ea{
	O_i^\dag \cdot U |_{H_i} \cdot O_i = O_0^\dag \cdot U |_{H_0} \cdot O_0 = \sum_{i,j\in [0..d-1]} c_{ij} S^i(C^2)^j
}
It's useful to note that $O_0$ can also be defined by $O_0 | i \ra = | i \ra | i \ra$, $\forall i \in [0..d-1]$.

It will be convenient for us to consider unitary operators $K_i \in \mcl{U}(H \otimes H)$ such that $K_i$ swaps subspaces $H_i$ and $H_0$ in accordance to $O_i, O_0$ and acts as identity on $(H_0 \oplus H_i)^\perp$. That is 
\ea{
	K_i |_{H_i} = O_0 \cdot O_i^\dag, 
	\nn
	K_i |_{H_0} = O_i \cdot O_0^\dag, 
	\nn
	K_i |_{(H_0 \oplus H_i)^\perp} = I. 
}

Clearly $K_i^2=I$ and $K_i^\dag = K_i$. Also 
\ea{
	K_iP_i = P_0K_i,
	\nn 
	P_iK_i = K_iP_0,
} 
where $P_i$ is the projector on $H_i$. 

Operators $K_i$ satisfy a useful property $\forall U = \sum_{i,j\in [0..d-1]} c_{ij} (S^{\otimes2})^i(C^{\otimes2})^j$: 
\ea{
	K_i U = U K_i
}

In fact, in the odd case it's not hard to find an explicit formula for $K_i$ (though we won't use it).
$\forall j \in [0..d-1], \forall i \in [1, .. , \frac{d-1}{2}]$: $K_0 = I$ and
\ea{
	K_i |j+\frac{i}{2}\ra|j+\frac{i}{2}\ra = \frac{1}{\sqrt{2}} (| j \ra | j+i \ra +  | j+i \ra | j \ra),
	\nn
	K_{-i} |j+\frac{i}{2}\ra|j+\frac{i}{2}\ra = \frac{1}{\sqrt{2}} (| j \ra | j+i \ra -  | j+i \ra | j \ra),
}
where the sum of indices and division by 2 is taken modulo $d$. 

We are ready to describe the set $\mbb{U}_{WH}$ in the odd case.

\begin{theorem}\label{tm:uwh-odd}
Assume $d$ is odd. Then 
\begin{enumerate}
\item
For any unitary operator $U_0 \in \mcl{U}(H_0)$ there exists unique $U \in \mbb{U}_{WH}$ such that
\ea{
	U = \bigoplus_{i \in [-\frac{d-1}{2}..\frac{d-1}{2}]} U |_{H_i} = \bigoplus_{i \in [-\frac{d-1}{2}..\frac{d-1}{2}]} K_i \cdot U_0 \cdot K_i |_{H_i} 
}
\item
Any operator $U \in \mbb{U}_{WH}$ has such form for a unitary $U_0 = U |_{H_0} \in \mcl{U}(H_0)$
\end{enumerate}
\end{theorem}

\begin{proof}
The second statement follows from the definitions of $O_i$ and $K_i$. To see why any unitary $U_0$ is suitable note that if $d$ is odd then operators $C^2$ and $S$ generate all $S^iC^j$, $(i,j) \in [0..d-1]^2$, up to some phase. So any linear operator on $H$ can be represented as a linear combination of $S^i(C^2)^j$, thus any unitary $U_0$ on $H_0$ can be represented as
\ea{
	U_0 = O_0 \cdot \sum_{i,j \in [0..d-1]} c_{ij}S^i(C^2)^j \cdot O_0^\dag
}
Hence 
\ea{
	U = \sum_{i,j \in [0..d-1]} c_{ij}(S^{\otimes 2})^i(C^{\otimes 2})^j 
}
is the required unitary from $\mbb{U}_{WH}$ because $U|_{H_0} = U_0$ and $U|_{H_i} = K_i \cdot U|_{H_0} \cdot K_i |_{H_i}$.  

\end{proof}

Finally, let's describe the set $\mbb{B}_{WH}$ in the odd case. 

As $P_i$ we denote the orthogonal projection on $H_i$ (hence $P_i \mbb{T}_a = \mbb{T}_a P_i$). 

\begin{theorem}\label{tm:bwh-odd} {(The description of the set $\mbb{B}_{WH}$ in the odd case)} \quad\\
Assume $d$ is odd. Let $|b\ra \in H \otimes H$ with $\la b | b \ra=1$. 
For $i \in [-\frac{d-1}{2}..\frac{d-1}{2}]$ define 
\ea{
	|b_i^0\ra = K_i P_i |b\ra,
} 
so $|b\ra = \sum_i K_i |b_i^0\ra$ and $|b_i^0\ra \in H_0$. Then $|b\ra \in \mbb{B}_{WH}$ if and only if 
\ea{
	\forall i,j \in [-\frac{d-1}{2}..\frac{d-1}{2}] ~:~ \la b_i^0 | b_j^0 \ra = \frac{1}{d} \delta_{ij}
}
In other words, to construct $|b\ra \in \mbb{B}_{WH}$ we can pick any orthogonal basis $\{ |b_i^0\ra \}$ in $H_0$ with the norm of each vector $\frac{1}{\sqrt{d}}$ and set $|b\ra = \sum_i K_i | b_i^0\ra$. 
\end{theorem}
\begin{proof}

One way to prove this is to combine theorems \ref{tm:uwh-group}, \ref{tm:uwh-odd} and prop. \ref{prop:0F0-bwh} as was originally done. Though there is a shorter proof. Let $|b\ra \in \mbb{B}_{WH}$ and $|b_i^0\ra = K_i P_i |b\ra$. For all $a \in [0..d-1]^2$, $a \neq (0,0)$ we have 
$$
0 = \langle b | \mbb{T}_a | b \rangle = \bigg(\sum_{i \in [-\frac{d-1}{2}..\frac{d-1}{2}]} \la b_i^0 | K_i \bigg) \cdot \mbb{T}_a \cdot \bigg(\sum_{j \in [-\frac{d-1}{2}..\frac{d-1}{2}]} K_j |b_j^0\ra \bigg) = 
$$
$$
= \bigg(\sum_{i \in [-\frac{d-1}{2}..\frac{d-1}{2}]} \la b_i^0 | K_i \bigg) \cdot \bigg(\sum_{j \in [-\frac{d-1}{2}..\frac{d-1}{2}]} K_j \mbb{T}_a |b_j^0\ra \bigg) = 
\sum_{i \in [-\frac{d-1}{2}..\frac{d-1}{2}]} \la b_i^0 | K_i K_i \mbb{T}_a |b_i^0\ra =
$$
\ea{
	= \sum_{i \in [-\frac{d-1}{2}..\frac{d-1}{2}]} \la b_i^0 | \mbb{T}_a |b_i^0\ra = \Tr \bigg( \mbb{T}_a \sum_{i \in [-\frac{d-1}{2}..\frac{d-1}{2}]} |b_i^0\ra \la b_i^0 | \bigg)
}

Let's denote 
\ea{
	B_0 = \sum_{i \in [-\frac{d-1}{2}..\frac{d-1}{2}]} |b_i^0\ra \la b_i^0 |
}

It's easy to see that $B_0=B_0^\dag$, $B_0 \geq 0$ and $B_0(H_0^\perp) = 0$. Note that 
\ea{
	\Tr(B_0) = \sum_{i \in [-\frac{d-1}{2}..\frac{d-1}{2}]}  \Tr \bigg( K_i P_i |b\ra \la b | P_i K_i \bigg) = 
	\Tr \bigg( |b\ra \la b | \cdot \sum_{i \in [-\frac{d-1}{2}..\frac{d-1}{2}]} P_i \bigg) = 1
}
Since $\Tr(\mbb{T}_a B_0) = 0$ for every $a \in [0..d-1]^2$, $a \neq (0,0)$, then $\forall (i,j) \in [0..d-1]^2$, $(i,j) \neq (0,0)$, we have  
$\Tr( O_0 \cdot S^iC^j \cdot O_0^\dag \cdot B_0 |_{H_0} ) = 0$. Hence $B_0 |_{H_0} = cI$ for some $c>0$. But $\Tr(B_0) = 1$ so $B_0 |_{H_0}  = \frac{1}{d}I$. It follows that vectors $|b_i^0\ra$ are orthogonal and have norm $\frac{1}{\sqrt{d}}$. 

Clearly, the reverse implication is also true, since for $|b\ra$ constructed from such $|b_i^0\ra \in H_0$ we have $\langle b | \mbb{T}_a | b \rangle = 0$, hence $|b\ra \in \mbb{B}_{WH}$. 

\end{proof}

Note that theorem \ref{tm:uwh-group} follows easily from this concrete description of $\mbb{B}_{WH}$ and from the description of $\mbb{U}_{WH}$ (theorem \ref{tm:uwh-odd}). 

It's also worth to note a simple corollary: the set $\mbb{B}_{WH}$ can be naturally parametrized by unitary matrices of size $d \times d$. 

\subsection{The case when $d$ is even} 

The general scheme is similar to the odd case, though this case is more technical. It was already noted that in this case $\{ C \otimes C, S \otimes S \}$ has 4 possible irreducible $\frac{d}{2}$-dimensional representations $\pi_{(1,1)} = \{ C_{\frac{d}{2}}, S_{\frac{d}{2}}\}$, $\pi_{(-1,1)} = \{ \omega C_{\frac{d}{2}}, S_{\frac{d}{2}}\}$, $\pi_{(1, -1)} = \{ C_{\frac{d}{2}}, \omega S_{\frac{d}{2}}\}$, $\pi_{(-1, -1)} = \{ \omega C_{\frac{d}{2}}, \omega S_{\frac{d}{2}}\}$ in its decomposition (here $C_{\frac{d}{2}}$ and $S_{\frac{d}{2}}$ are clock and shift matrices of size $\frac{d}{2} \times \frac{d}{2}$). In fact, every $\pi_{(\pm1,\pm1)}$ appears exactly $\frac{d}{2}$ times in the decomposition of $\{ C \otimes C, S \otimes S \}$. Below we present one such explicit decomposition. 


For $s=0$ denote

\ea{
	H_{0}^{(1,1)} = span\{ |j\ra|j\ra + |j+\frac{d}{2}\ra|j+\frac{d}{2}\ra\}_{ j\in [0..\frac{d}{2}-1] },
}
\ea{
	H_{0}^{(1,-1)} = span\{ |j\ra|j\ra - |j+\frac{d}{2}\ra|j+\frac{d}{2}\ra\}_{ j\in [0..\frac{d}{2}-1] }
}

and for $s \in [1..\frac{d}{2}-1]$ denote $\alpha = (-1)^s$,
\ea{
	H_{s}^{(\alpha,1)} = span\{ |j\ra|j+s\ra + |j+s\ra|j\ra + 
	\nn
	+ |j+\frac{d}{2}\ra|j+s+\frac{d}{2}\ra + |j+s+\frac{d}{2}\ra|j+\frac{d}{2}\ra \}_{ j\in [0..\frac{d}{2}-1] },
}
\ea{
	H_{s}^{(\alpha,-1)} = span\{ |j\ra|j+s\ra + |j+s\ra|j\ra -
	\nn
	- |j+\frac{d}{2}\ra|j+s+\frac{d}{2}\ra - |j+s+\frac{d}{2}\ra|j+\frac{d}{2}\ra \}_{ j\in [0..\frac{d}{2}-1] },
}

also 

\ea{
	H_{-s}^{(\alpha,1)} = span\{ |j\ra|j+s\ra - |j+s\ra|j\ra +
	\nn
	+ |j+\frac{d}{2}\ra|j+s+\frac{d}{2}\ra - |j+s+\frac{d}{2}\ra|j+\frac{d}{2}\ra \}_{ j\in [0..\frac{d}{2}-1] },
}
\ea{
	H_{-s}^{(\alpha,-1)} = span\{ |j\ra|j+s\ra - |j+s\ra|j\ra - 
	\nn
	- |j+\frac{d}{2}\ra|j+s+\frac{d}{2}\ra + |j+s+\frac{d}{2}\ra|j+\frac{d}{2}\ra \}_{ j\in [0..\frac{d}{2}-1] }.
}

Finally, for $s=\frac{d}{2}$ denote $\alpha = (-1)^s = (-1)^\frac{d}{2}$,
\ea{
	H_{\frac{d}{2}}^{(\alpha,1)} = span\{ |j\ra|j+\frac{d}{2}\ra + |j+\frac{d}{2}\ra|j\ra\}_{ j\in [0..\frac{d}{2}-1] },
}
\ea{
	H_{-\frac{d}{2}}^{(\alpha,-1)} = span\{ |j\ra|j+\frac{d}{2}\ra - |j+\frac{d}{2}\ra|j\ra\}_{ j\in [0..\frac{d}{2}-1] }.
}
Note that $H_{s}^{(\alpha,\beta)}$ are defined not for every possible combination of $(\alpha, \beta)=(\pm 1, \pm 1)$ and $s \in [-\frac{d}{2},..,\frac{d}{2}]$.

It's a straightforward routine to check that 
\begin{enumerate}
\item All defined $H_s^{(\alpha, \beta)}$ ($2d$ of them) have dimension $\frac{d}{2}$ 
\item $H \otimes H$ is the direct sum of all $H_s^{(\alpha, \beta)}$
\item $H_{sym}$ is the direct sum of all $H_s^{(\alpha, \beta)}$ where $s\geq 0$, ($d+1$ of them)
\item Each $H_s^{(\alpha, \beta)}$ is invariant subspace for operators $C\otimes C$ and $S \otimes S$ 
\item For each subspace $\mcl{H}=H_s^{(\alpha, \beta)}$: 
\ea{
	(C\otimes C)^\frac{d}{2} |_\mcl{H} = \alpha I,
}
\ea{
	(S\otimes S)^\frac{d}{2} |_\mcl{H} = \beta I
}

\item For each of the 4 possible values $(\alpha, \beta)=(\pm 1, \pm 1)$ there are exactly $\frac{d}{2}$ subspaces of the type $H_s^{(\alpha, \beta)}$

\end{enumerate}
It follows that $\forall U = \sum_{i,j\in [0..d-1]} c_{ij} (S^{\otimes2})^i(C^{\otimes2})^j$:
\ea{
	U = \bigoplus U |_{ H_s^{(\alpha, \beta)}}
}
Again we use Theorem VII.5.1 from \cite{Davidson} to conclude that for every $\mcl{H} = H_s^{(\alpha, \beta)}$ there is an isometry $O_s^{(\alpha, \beta)} : \mbb{C}^\frac{d}{2} \rightarrow \mcl{H}$ such that

\ea{
	(O_s^{(\alpha, \beta)})^\dag \cdot (C\otimes C) |_\mcl{H} \cdot O_s^{(\alpha, \beta)} = \omega^{ \frac{1-\alpha}{2}} C_{\frac{d}{2}} \text{ on } \mbb{C}^\frac{d}{2},
}
\ea{
	(O_s^{(\alpha, \beta)})^\dag \cdot (S\otimes S) |_\mcl{H} \cdot O_s^{(\alpha, \beta)} = \omega^{ \frac{1-\beta}{2}} S_{\frac{d}{2}} \text{ on } \mbb{C}^\frac{d}{2}.
}

Now for every $(\alpha, \beta)=(\pm 1, \pm 1)$ let's mark a single subspace of the type $H_s^{(\alpha, \beta)}$ and denote it as $H_\bullet^{(\alpha, \beta)}$ (in the odd case we've just used $H_0$ as $H_\bullet$). We also denote
\ea{
	H_\bullet = H_\bullet^{(1, 1)} \oplus H_\bullet^{(-1, 1)} \oplus H_\bullet^{(1, -1)} \oplus H_\bullet^{(-1, -1)}
}

Let's define unitary operators $K_s^{(\alpha, \beta)} \in \mcl{U}(H \otimes H)$ that swap subspaces $H_s^{(\alpha, \beta)}$ and $H_\bullet^{(\alpha, \beta)}$ in accordance to $O_s^{(\alpha, \beta)}$, $O_\bullet^{(\alpha, \beta)}$:
\ea{
	K_s^{(\alpha, \beta)} |_{H_s^{(\alpha, \beta)}} = O_\bullet^{(\alpha, \beta)} \cdot (O_s^{(\alpha, \beta) })^\dag ,
	\nn
	K_s^{(\alpha, \beta)} |_{H_\bullet^{(\alpha, \beta)}} = O_s^{(\alpha, \beta)} \cdot (O_\bullet^{(\alpha, \beta)})^\dag ,
	\nn
	K_s^{(\alpha, \beta)} |_{(H_\bullet^{(\alpha, \beta)} \oplus H_s^{(\alpha, \beta)} )^\perp} = I. ~~~~~~~~~~~
}
Again we have that $(K_s^{(\alpha, \beta)})^2=I$, $(K_s^{(\alpha, \beta)})^\dag=K_s^{(\alpha, \beta)}$. 

Also for every $(\alpha, \beta) = (\pm1,\pm1)$ and every $s$:
\ea{
	K_s^{(\alpha, \beta)}P_s^{(\alpha, \beta)} = P_\bullet^{(\alpha, \beta)}K_s^{(\alpha, \beta)},
	\nn
	P_s^{(\alpha, \beta)}K_s^{(\alpha, \beta)} = K_s^{(\alpha, \beta)}P_\bullet^{(\alpha, \beta)},
}
where $P_s^{(\alpha, \beta)}$ is the projector on $H_s^{(\alpha, \beta)}$.

And for any $\forall U = \sum_{i,j\in [0..d-1]} c_{ij} (S^{\otimes2})^i(C^{\otimes2})^j$: 
\ea{
	K_s^{(\alpha, \beta)} U = U K_s^{(\alpha, \beta)}.
}

\begin{theorem}\label{tm:uwh-even}
Assume $d$ is even. Then 
\begin{enumerate}
\item
For any 4 unitary operators $U_{(\pm1,\pm1)} \in \mcl{U}(H_\bullet^{(\pm1,\pm1)})$ there exists a unique $U \in \mbb{U}_{WH}$ such that
\ea{
	U = \bigoplus_{ (\alpha,\beta)=(\pm1,\pm1)} \Bigg( \bigoplus_{s} K_s^{(\alpha,\beta)} \cdot U_{(\alpha,\beta)} \cdot K_s^{(\alpha,\beta)} |_{H_s^{(\alpha,\beta)}} \Bigg)
}
\item
Any operator $U \in \mbb{U}_{WH}$ has such form for 4 unitaries 
\ea{
	U_{(\pm1,\pm1)} = U |_{H_\bullet^{(\pm1,\pm1)}} \in \mcl{U}(H_\bullet^{(\pm1,\pm1)})
}
\end{enumerate}
\end{theorem}

\begin{proof}
Similar to the odd case. The only thing needs to be explained is why 4 parts of $U \in \mbb{U}_{WH}$:
\ea{
	U |_{H_\bullet^{(1,1)}}, U |_{H_\bullet^{(-1,1)}}, U |_{H_\bullet^{(1,-1)}}, U |_{H_\bullet^{(-1,-1)}} 
} 
can be any 4 unitaries on the respective subspaces. We have that $\forall (\alpha,\beta) = (\pm1,\pm1)$, $c_{ij} \in \mbb{C}, i,j \in [0..d-1]$:
\ea{
	c_{ij}(S^{\otimes2})^i(C^{\otimes2})^j |_{H_\bullet^{(\alpha,\beta)}} 
= O_\bullet^{(\alpha, \beta)} \cdot c_{ij} (\omega^{\frac{1-\beta}{2}})^i (\omega^{\frac{1-\alpha}{2}})^j S_{\frac{d}{2}}^i C_{\frac{d}{2}}^j \cdot (O_\bullet^{(\alpha, \beta)})^\dag
}

Thus for $i,j \in [0..\frac{d}{2}-1]$:
$$
\bigg( c_{i,j}(S^{\otimes2})^i(C^{\otimes2})^j + c_{i+\frac{d}{2}, j}(S^{\otimes2})^{i+\frac{d}{2}}(C^{\otimes2})^j + 
$$
$$
+ c_{i,j+\frac{d}{2}}(S^{\otimes2})^i(C^{\otimes2})^{j+\frac{d}{2}} + c_{i+\frac{d}{2},j+\frac{d}{2}}(S^{\otimes2})^{i+\frac{d}{2}}(C^{\otimes2})^{j+\frac{d}{2}} \bigg) |_{H_\bullet^{(\alpha,\beta)}} =
$$
$$
= \bigg( c_{i,j} (\omega^{\frac{1-\beta}{2}})^i (\omega^{\frac{1-\alpha}{2}})^j + c_{i+\frac{d}{2},j} (\omega^{\frac{1-\beta}{2}})^{i+\frac{d}{2}} (\omega^{\frac{1-\alpha}{2}})^j + 
$$
$$
+ c_{i,j+\frac{d}{2}} (\omega^{\frac{1-\beta}{2}})^i (\omega^{\frac{1-\alpha}{2}})^{j+\frac{d}{2}} + c_{i+\frac{d}{2},j+\frac{d}{2}} (\omega^{\frac{1-\beta}{2}})^{i+\frac{d}{2}} (\omega^{\frac{1-\alpha}{2}})^{j+\frac{d}{2}} \bigg) 
\times 
$$
\ea{
	\times O_\bullet^{(\alpha, \beta)} \cdot S_{\frac{d}{2}}^i C_{\frac{d}{2}}^j \cdot (O_\bullet^{(\alpha, \beta)})^\dag = \chi_{i,j}^{(\alpha, \beta)} \cdot O_\bullet^{(\alpha, \beta)} \cdot S_{\frac{d}{2}}^i C_{\frac{d}{2}}^j \cdot (O_\bullet^{(\alpha, \beta)})^\dag 
}
Here the complex number $\chi_{i,j}^{(\alpha, \beta)}$ is just the corresponding expression that depends on
$c_{i,j}$, $c_{i+\frac{d}{2},j}$, $c_{i,j+\frac{d}{2}}$, $c_{i+\frac{d}{2},j+\frac{d}{2}}$. 
For any $i,j \in [0..\frac{d}{2}-1]$ four vectors $v_{ij}^{(\pm1,\pm1)} \in \mbb{C}^4$, where  
\ea{
	v_{ij}^{(\alpha,\beta)} = \bigg( (\omega^{\frac{1-\beta}{2}})^i (\omega^{\frac{1-\alpha}{2}})^j, (\omega^{\frac{1-\beta}{2}})^{i+\frac{d}{2}} (\omega^{\frac{1-\alpha}{2}})^j, 
	\nn
	(\omega^{\frac{1-\beta}{2}})^i (\omega^{\frac{1-\alpha}{2}})^{j+\frac{d}{2}}, (\omega^{\frac{1-\beta}{2}})^{i+\frac{d}{2}} (\omega^{\frac{1-\alpha}{2}})^{j+\frac{d}{2}} \bigg),
}
are linearly independent. Hence $\forall i,j \in [0..\frac{d}{2}-1]$ for any 4 complex numbers $\chi_{i,j}^{(\pm1, \pm1)}$ we can always find 4 corresponding numbers $c_{i,j},c_{i+\frac{d}{2},j},c_{i,j+\frac{d}{2}},c_{i+\frac{d}{2},j+\frac{d}{2}}$.  

To see why $v_{ij}^{(1,1)}$, $v_{ij}^{(1,-1)}$,$v_{ij}^{(-1,1)}$,$v_{ij}^{(-1,-1)}$ are linearly independent consider the corresponding matrix 

\ea{
	M_{ij} =  \bmt v_{ij}^{(1,1)} \\ v_{ij}^{(1,-1)} \\ v_{ij}^{(-1,1)} \\ v_{ij}^{(-1,-1)} \emt 
= \bmt 
1 & 1 & 1 & 1 \\ 
\omega^i & -\omega^i & \omega^i & -\omega^i \\ 
\omega^j & \omega^j & -\omega^j & -\omega^i \\
\omega^{i+j} & -\omega^{i+j} & -\omega^{i+j} & \omega^{i+j} \\
\emt
}
 
This matrix has the same rank as the matrix 
\ea{
	\bmt 
1 & 1 & 1 & 1 \\ 
1 & -1 & 1 & -1 \\ 
1 & 1 & -1 & -1 \\
1 & -1 & -1 & 1 \\
\emt
}
which has det = 16.   

\end{proof}

Let's describe the set $\mbb{B}_{WH}$ in the even case. 

As $P_s^{(\alpha,\beta)}$ we denote the ortho-projection on $H_s^{(\alpha,\beta)}$ (hence $P_s^{(\alpha,\beta)} \mbb{T}_a = \mbb{T}_a P_s^{(\alpha,\beta)}$). 

\begin{theorem}\label{tm:bwh-even} {(The description of the set $\mbb{B}_{WH}$ in the even case)} \quad\\
Assume $d$ is even. Let $|b\ra \in H \otimes H$ with $\la b | b \ra=1$. 
For $(\alpha,\beta) = (\pm1, \pm1)$ and $s \in [-\frac{d}{2},..,\frac{d}{2}]$ define 
\ea{
	|b_s^{(\alpha,\beta)}\ra = K_s^{(\alpha,\beta)} P_s^{(\alpha,\beta)} |b\ra,
}
so $|b_s^{(\alpha,\beta)}\ra \in H_\bullet^{(\alpha,\beta)}$ and $|b\ra = \sum K_s^{(\alpha,\beta)} |b_s^{(\alpha,\beta)}\ra$. Then $|b\ra \in \mbb{B}_{WH}$ if and only if 
\ea{
	\forall (\alpha,\beta) = (\pm1, \pm1) ~~~ \forall s,t \in [-\frac{d}{2},..,\frac{d}{2}] ~:~ \la b_s^{(\alpha,\beta)} | b_t^{(\alpha,\beta)} \ra = \frac{1}{2d} \delta_{st}
}
i.e. to construct $|b\ra \in \mbb{B}_{WH}$ we can pick any orthonormal basis from $H_\bullet^{(\alpha,\beta)}$ for every $(\alpha,\beta) = (\pm1, \pm1)$. 
\end{theorem}
\begin{proof}
Similarly to the odd case we can define operator $B_\bullet$:
$$
B_\bullet = \sum_{(\alpha,\beta) = (\pm1, \pm1)} \sum_s |b_s^{(\alpha,\beta)}\ra \la b_s^{(\alpha,\beta)} |
$$
and prove that $\Tr(\mbb{T}_a B_\bullet)=0$ for $\forall a\neq(0,0)$. Hence $B_\bullet$ is equal to $\frac{1}{2d}I$ on $H_\bullet$.

\end{proof}

It follows that in this case $\mbb{B}_{WH}$ can be parametrized by 4 unitary matrices of size $\frac{d}{2} \times \frac{d}{2}$. 

\section{SIC-POVM tensor square as a projection of a WH-type basis}
\label{sec:main-tm}

Now we are ready to prove our main result

\begin{theorem}\label{tm:main} (Main Theorem) 

Let $|f\ra \in H$ be a fiducial vector of a SIC of WH-type. Then there exists vector $|b\ra \in \mbb{B}_{WH}$ such that 
\ea{
	|f\ra|f\ra = \sqrt{ \frac{2d}{d+1} } P_{sym}|b\ra
} 
In other words, $\{ \mbb{T}_i |f\ra|f\ra \}_{i \in [0..d-1]^2}$ is the scaled projection of the basis $\{ \mbb{T}_i |b\ra \}_{i \in [0..d-1]^2}$ onto $H_{sym}$.
\end{theorem}

\begin{proof} \quad\\

1) The case when $d$ is odd. 

First of all, note that 

\ea{
	|f\ra|f\ra =  \sum_{i \in [0..\frac{d-1}{2}]} P_i |f\ra|f\ra
}

For $i \in [0..\frac{d-1}{2}]$ let's define
\ea{
	|b_i^0\ra = \sqrt{ \frac{d+1}{2d} } K_i P_i |f\ra|f\ra
}

Clearly, $|b_i^0\ra \in H_0$, and we can write 

\ea{
	|f\ra|f\ra = \sqrt{ \frac{2d}{d+1} } \sum_{i \in [0..\frac{d-1}{2}]} K_i |b_i^0\ra
}

What we are going to prove is that $\forall i \in [0..\frac{d-1}{2}]$ vectors $|b_i^0\ra$ are orthogonal to each other and have norm $\frac{1}{\sqrt{d}}$. Therefore we can complete them to a full basis on $H_0$ by picking any suitable $|b_i^0\ra$ for $i \in [-\frac{d-1}{2},..,-1]$. Then for a vector $|b\ra = \sum_{i \in [-\frac{d-1}{2}..\frac{d-1}{2}]} K_i | b_i^0\ra \in \mbb{B}_{WH}$ we will have the required identity 
\ea{
	|f\ra|f\ra = \sqrt{ \frac{2d}{d+1} } P_{sym}|b\ra
}

Let's verify. 

Recall that $P_i\mbb{T}_a = \mbb{T}_aP_i$ and $K_i\mbb{T}_a = \mbb{T}_aK_i$ for any $i \in [-\frac{d-1}{2}..\frac{d-1}{2}]$, $a \in [0..d-1]^2$. 

For $i \in [0..\frac{d-1}{2}]$ denote 
\ea{
	\mu_i = \la b_i^0 | b_i^0 \ra = \frac{d+1}{2d} \la f|\la f| P_i K_i K_i P_i |f\ra|f\ra = \frac{d+1}{2d} \Tr\Big(P_i \cdot |f\ra|f\ra \la f|\la f| \Big)
}

For any $a \in [0..d-1]^2$ we also have 
\ea{
	\mu_i = \frac{d+1}{2d} \Tr\Big(P_i \cdot \mbb{T}_a |f\ra|f\ra \la f|\la f| \mbb{T}_{-a} \Big)
}
So after summation we have 
$$
d^2 \mu_i = \frac{d+1}{2d} \Tr\Big(P_i \cdot \sum_{a \in [0..d-1]^2} \mbb{T}_a |f\ra|f\ra \la f|\la f| \mbb{T}_{-a} \Big) = 
$$
\ea{
	= \frac{d+1}{2d} \Tr\Big( P_i \cdot \frac{2 d^2}{d(d+1)} P_{sym} \Big) = \Tr\Big( P_i \cdot P_{sym} \Big) = \Tr\Big( P_i \Big) = d
}
hence $\mu_i = \frac{1}{d}$. 

In a similar way we can prove orthogonality. 

For $i,j \in [0..\frac{d-1}{2}]$, $i \neq j$, denote 
\ea{
	\mu_{ij} = \la b_i^0 | b_j^0 \ra = \frac{d+1}{2d} \la f|\la f| P_i K_i K_j P_j |f\ra|f\ra = \frac{d+1}{2d} \Tr\Big(P_i K_i K_j P_j \cdot |f\ra|f\ra \la f|\la f| \Big)
}
For $a \in [0..d-1]^2$ we also have 
\ea{
	\mu_{ij} = \frac{d+1}{2d} \Tr\Big(P_i K_i K_j P_j \cdot \mbb{T}_a |f\ra|f\ra \la f|\la f| \mbb{T}_{-a} \Big)
}
So 
$$
d^2 \mu_{ij} = \frac{d+1}{2d} \Tr\Big(P_i K_i K_j P_j \cdot \sum_{a \in [0..d-1]^2} \mbb{T}_a |f\ra|f\ra \la f|\la f| \mbb{T}_{-a} \Big) = 
$$
\ea{
	= \frac{d+1}{2d} \Tr\Big( P_i K_i K_j P_j \cdot \frac{2 d^2}{d(d+1)} P_{sym} \Big) = \Tr\Big( K_i K_j \cdot P_j P_{sym} P_i \Big) = 0
}
hence $\mu_{ij} = 0$. 

\quad

2) The case when $d$ is even. 

Similarly, we have that 
\ea{
	|f\ra|f\ra =  \sum_{s\geq 0}^{(\alpha, \beta)=(\pm1,\pm1)} P_s^{(\alpha, \beta)} |f\ra|f\ra
}
For $(\alpha, \beta)=(\pm1,\pm1)$ and $s \geq 0$ define 
\ea{
	|b_s^{(\alpha, \beta)}\ra = \sqrt{ \frac{d+1}{2d} } K_s^{(\alpha, \beta)} P_s^{(\alpha, \beta)} |f\ra|f\ra \in H_\bullet^{(\alpha, \beta)}
}

Let's verify that $\forall (\alpha, \beta)=(\pm1,\pm1)$ vectors $|b_s^{(\alpha, \beta)}\ra$, $s\geq0$, are orthogonal to each other and have norm $\frac{1}{\sqrt{2d}}$.

For $s \in [0..\frac{d}{2}]$ denote 
$$
\mu_s = \la b_s^{(\alpha, \beta)} | b_s^{(\alpha, \beta)} \ra = \frac{d+1}{2d} \la f|\la f| P_s^{(\alpha, \beta)} |f\ra|f\ra = 
$$
\ea{	
	= \frac{d+1}{2d} \Tr\Big(P_s ^{(\alpha, \beta)} \cdot |f\ra|f\ra \la f|\la f| \Big)
}

For any $a \in [0..d-1]^2$ we also have 
\ea{
	\mu_s = \frac{d+1}{2d} \Tr\Big(P_s^{(\alpha, \beta)} \cdot \mbb{T}_a |f\ra|f\ra \la f|\la f| \mbb{T}_{-a} \Big)
}
So after summation
$$
d^2 \mu_s = \frac{d+1}{2d} \Tr\Big(P_s^{(\alpha, \beta)} \cdot \sum_{a \in [0..d-1]^2} \mbb{T}_a |f\ra|f\ra \la f|\la f| \mbb{T}_{-a} \Big) = 
$$
\ea{
	= \frac{d+1}{2d} \Tr\Big( P_s^{(\alpha, \beta)} \cdot \frac{2 d^2}{d(d+1)} P_{sym} \Big) = \Tr\Big( P_s^{(\alpha, \beta)} \cdot P_{sym} \Big) = \Tr\Big( P_s^{(\alpha, \beta)} \Big) = \frac{d}{2}
}
hence $\mu_s = \frac{1}{2d}$. 

For $s,t \in [0..\frac{d}{2}]$, $s \neq t$, denote 
$$
\mu_{st} = \la b_s^{(\alpha, \beta)} | b_t^{(\alpha, \beta)} \ra = \frac{d+1}{2d} \la f|\la f| P_s^{(\alpha, \beta)} K_s^{(\alpha, \beta)} K_t^{(\alpha, \beta)} P_t^{(\alpha, \beta)} |f\ra|f\ra = 
$$
\ea{
	= \frac{d+1}{2d} \Tr\Big(P_s^{(\alpha, \beta)} K_s^{(\alpha, \beta)} K_t^{(\alpha, \beta)} P_t^{(\alpha, \beta)} \cdot |f\ra|f\ra \la f|\la f| \Big)
}
For $a \in [0..d-1]^2$ we also have 
\ea{
	\mu_{st} = \frac{d+1}{2d} \Tr\Big(P_s^{(\alpha, \beta)} K_s^{(\alpha, \beta)} K_t^{(\alpha, \beta)} P_t^{(\alpha, \beta)} \cdot \mbb{T}_a |f\ra|f\ra \la f|\la f| \mbb{T}_{-a} \Big)
}
So 
$$
d^2 \mu_{st} = \frac{d+1}{2d} \Tr\Big(P_s^{(\alpha, \beta)} K_s^{(\alpha, \beta)} K_t^{(\alpha, \beta)} P_t^{(\alpha, \beta)} \cdot \sum_{a \in [0..d-1]^2} \mbb{T}_a |f\ra|f\ra \la f|\la f| \mbb{T}_{-a} \Big) = 
$$
$$
= \frac{d+1}{2d} \Tr\Big( P_s^{(\alpha, \beta)} K_s^{(\alpha, \beta)} K_t^{(\alpha, \beta)} P_t^{(\alpha, \beta)} \cdot \frac{2 d^2}{d(d+1)} P_{sym} \Big) = 
$$
\ea{
	= \Tr\Big( K_s^{(\alpha, \beta)} K_t^{(\alpha, \beta)} \cdot P_t^{(\alpha, \beta)} P_{sym} P_s^{(\alpha, \beta)} \Big) = 0
}
hence $\mu_{st} = 0$. Also, clearly, $| b_s^{(\alpha_1, \beta_1)} \ra$ is orthogonal to $| b_t^{(\alpha_2, \beta_2)} \ra $ if $(\alpha_1, \beta_1) \neq (\alpha_2, \beta_2)$. 

Similarly to the odd case, for every $(\alpha, \beta)=(\pm1,\pm1)$ we can complete the set $\{ |b_s^{(\alpha, \beta)}\ra \}_{s\geq 0}$ to a full basis on $H_\bullet^{(\alpha, \beta)}$ and construct $|b\ra \in \mbb{B}_{WH}$ that satisfy theorem requirement.
\end{proof}

\section{WH-type bases as a search space for SIC solutions}
\label{sec:search-space}
Theorem \ref{tm:main} suggests that one could use WH-type bases as a search space for SIC solutions of WH-type. To find SIC via this way the only thing needs to be satisfied is the Schmidt rank of the corresponding vector:
\begin{proposition}\label{prop:bwh-search}
Let $|b\ra \in \mbb{B}_{WH}$. If $srank(P_{sym}|b\ra)=1$, that is $P_{sym}|b\ra = \lambda |f\ra|f\ra$ for some unit-norm $|f\ra \in H$ and $\lambda>0$, then $|f\ra$ is a WH-type SIC fiducial vector (and $\lambda$ must be $\sqrt{ \frac{d+1}{2d} }$).
\end{proposition}
\begin{proof}
By prop. \ref{prop:bwh-sym-proj} we have that $\{ \sqrt{ \frac{2d}{d+1} } \lambda \mbb{T}_a |f\ra|f\ra \}_{a\in[0..d-1]^2}$ is the unit-norm tight frame on $H_{sym}$. Hence $\lambda=\sqrt{ \frac{d+1}{2d} }$ and from theorem \ref{tm:sic-sym-crit} we deduce that $|f\ra$ must be a SIC fiducial. 
\end{proof}

Another argument to study $\mbb{B}_{WH}$ as a search space is that for the pretty simple element $|0\ra F|0\ra \in \mbb{B}_{WH}$ we already have $srank(P_{sym}|0\ra F|0\ra) = 2$. It also generates ETF (in general for $|b\ra \in \mbb{B}_{WH}$ the tight frame $\{ \mbb{T}_aP_{sym}|b\ra \}_{a\in[0..d-1]^2}$ is not equiangular). 

Computationally, the Schmidt rank of the density operator $\rho$ on $H \otimes H$ (i.e. $\rho=\rho^\dag\geq0$, $\Tr(\rho)=1$) can be computed as the rank of the partial trace $\rho_2 = \text{Tr}_1(\rho)$. Since $\Tr(\rho_2)=1$ one could try to maximize the function $\Tr( \rho_2^2 )$ over $\rho_2 = \text{Tr}_1(\rho)$, where $\rho = \frac{2d}{d+1}P_{sym}|b\ra\la b| P_{sym}$, $|b\ra \in \mbb{B}_{WH}$. The maximal value $\Tr( \rho_2^2 )=1$ would mean that a SIC solution is found. 

Unfortunately, our initial computational experiments (performed via gradient descent methods from \textit{Optim.jl} package in the Julia programming language) showed that one can stuck in a local maxima of $\Tr( \rho_2^2 )$. Though, clearly, this method must converge to a SIC if you are already close enough to a SIC solution (in the natural metric). 

\section{SIC-related symmetric tight fusion frames}
\label{sec:STFF}

Assume $d$ is odd and $|f\ra \in H$ is a fiducial vector of a WH-type SIC. For $a\in[0..d-1]^2$ let $\theta_a \in \mbb{R}$ describe the phases of scalar products between SIC vectors, that is 
\ea{
e^{i\theta_{(0,0)}} = \la f | f \ra = 1, \text{ so } \theta_{(0,0)} = 0,
\nn
	e^{i\theta_a} = \la f | T_{-a} |f\ra \cdot \sqrt{d+1}
}

In \cite{STFF}, theorem 7, it was discovered that for any matrix $X \in \text{GL}(2, \mbb{Z}/d\mbb{Z})$ with $\text{det}X=2^{-1} (\text{mod} ~ d)$ the following operator 
\ea{
	\Pi^{+} = \frac{d+1}{2d}I_d + \frac{1}{2d}\sum_{a\in[0..d-1]^2}^{a\neq(0,0)}e^{2i\theta_{Xa}}T_a
}
is the fiducial projector of rank $(d+1)/2$ of a symmetric tight fusion frame (of WH-type). That is if we denote $\Pi^{+}_a = T_a\Pi^{+}T_{-a}$ for $a \in [0..d-1]^2$ then 
\ea{
	(\Pi^{+}_a)^2 = (\Pi^{+}_a)^\dag = \Pi^{+}_a, ~~~ \Tr(\Pi^{+}_a) = \frac{d+1}{2},
}
\ea{
	\Tr(\Pi_a^{+}\Pi_b^{+}) = \frac{1}{4}(d+2 + d \delta_{ab}),
} 
\ea{
	\sum_{a\in[0..d-1]^2}\Pi_a^{+} = \frac{d(d+1)}{2} I_d.
}

Also, clearly, $\Pi^{-} = I - \Pi^{+}$ generates WH-type STFF of rank $(d-1)/2$. 

In the proof of theorem \ref{tm:main} we saw that the vectors $|b^0_i\ra = \sqrt{ \frac{d+1}{2d} } K_i P_i |f\ra|f\ra$, $i \in [0..\frac{d-1}{2}]$, are orthogonal to each other, have norm $\frac{1}{\sqrt{d}}$ and all lay in $H_0$. So there is the natural projector 
\ea{
	B^{+} = d \sum_{i \in [0..\frac{d-1}{2}]} |b^0_i\ra \la b^0_i |
}
of rank $(d+1)/2$ with the image in $H_0$. It's reasonable to expect that $B^{+} |_{H_0}$ coincides with $\Pi^{+}$ (under $O_0$ isometry defined by $O_0 | i \ra = | i \ra| i \ra$, $i \in [0..d-1]$) for some $X \in \text{GL}(2, \mbb{Z}/d\mbb{Z})$ with $\text{det}X=2^{-1}$. Indeed, let's check:

$$
\Tr \big( B^+ \mbb{T}_c \big) = \Tr \bigg( d \sum_{i \in [0..\frac{d-1}{2}]} |b^0_i\ra \la b^0_i |  \mbb{T}_c \bigg) 
= d \sum_{i \in [0..\frac{d-1}{2}]} \Tr \bigg( \la b^0_i |  \mbb{T}_c |b^0_i\ra  \bigg) = 
$$
$$
= d \cdot \frac{d+1}{2d} \sum_{i \in [0..\frac{d-1}{2}]} \Tr \bigg( \la f | \la f | P_i K_i \mbb{T}_c K_i P_i | f \ra | f \ra  \bigg) 
=
$$
$$
= \frac{d+1}{2} \sum_{i \in [0..\frac{d-1}{2}]} \Tr \bigg( \la f | \la f | P_i  \mbb{T}_c | f \ra | f \ra  \bigg) =
$$
\ea{
	= \frac{d+1}{2} \Tr \bigg( \la f | \la f | P_{sym}  \mbb{T}_c | f \ra | f \ra  \bigg) 
= \frac{d+1}{2} \big( \la f | T_c | f \ra \big)^2
}
Recall that $\forall c=(c_1,c_2) \in [0..d-1]^2$:
\ea{
	\mbb{T}_c |_{H_0} = \tau^{2c_1c_2} (S^{\otimes 2})^{c_1} (C^{\otimes 2})^{c_2} |_{H_0} = 
	\nn
= O_0 \cdot \tau^{2c_1c_2} S^{c_1} C^{2c_2} \cdot O_0^\dag 
= O_0 \cdot T_{(c_1,2c_2)} \cdot O_0^\dag
}
Hence
\ea{
	\Tr \big( O_0^\dag \cdot B^+ |_{H_0} \cdot O_0 \cdot T_{(c_1,2c_2)} \big) = \frac{d+1}{2} \big( \la f | T_c | f \ra \big)^2
}

On the other hand, for  $X = diag\{1,2^{-1}(\text{mod} ~ d)\} \in \text{GL}(2, \mbb{Z}/d\mbb{Z})$ and $b = (c_1,2c_2) \neq (0,0)$ we have 
$$
\Tr \big( \Pi^+ T_b \big)  = \Tr \bigg( \frac{d+1}{2d}T_b + \frac{1}{2d}\sum_{a\in[0..d-1]^2}^{a\neq(0,0)}e^{2i\theta_{Xb}}T_a  T_b \bigg)
=
$$ 
\ea{
	= d \cdot \frac{1}{2d} e^{2i\theta_{-Xb}} = \frac{d+1}{2} \big( \la f | T_c | f \ra \big)^2
}
Hence $\Pi^{+} = O_0^\dag \cdot B^{+} |_{H_0} \cdot  O_0$ for $X = diag\{1,2^{-1}(\text{mod} ~ d)\}$.

\section{Conclusions} 
\label{conc}
In this work we've investigated the properties of a WH-type SIC tensor square. We've proved that in any case such squared frame can be obtained as a projection of a WH-type basis of $H\otimes H$ onto the symmetric subspace. We gave a full description of the set of all WH-type bases and showed that it's reasonable to use this set as a search space for SIC solutions. Also we showed that a particular element of this set is close to a SIC solution in some structural sense. In the process of our investigation we've found a geometric construction of a SIC-related symmetric tight fusion frames.



%
%



\end{document}